\documentclass[11pt,a4paper,oneside]{article}

\usepackage{setspace}

\usepackage[top=3cm, bottom=3cm, left=2cm, right=2cm]{geometry}
\usepackage{authblk}
\usepackage{dirtytalk}

\usepackage{amsmath}
\usepackage{graphicx,psfrag,epsf}
\usepackage{enumerate}
\usepackage{natbib}
\usepackage{xcolor}
\usepackage{url} 
\usepackage{amsfonts}
\usepackage{amsthm}
\usepackage{amssymb}
\usepackage{booktabs}

\usepackage{mathrsfs} 
\usepackage{xcolor}
\usepackage[caption = false, labelformat=simple]{subfig}
\usepackage{bm}
\usepackage{bbm}
\usepackage{mathtools} 
\usepackage{comment}
\usepackage{dirtytalk}
\usepackage[backref=false]{hyperref}

\newtheorem{corollary}{Corollary}
\newtheorem{proposition}{Proposition}
\newtheorem*{remark}{Remark}

\newtheorem{definition}{Definition}

\def\simiid{\stackrel{\text {iid}}{\sim}}
\def\Pr{\mathrm{Pr}}
\def\CER{\mathrm{CER}}
\def\GCER{\mathrm{GCER}}

\DeclareMathOperator*{\argmax}{argmax}
\DeclareMathOperator*{\argmin}{argmin}

\newcommand{\Autoref}[1]{%
  \begingroup%
  \def\chapterautorefname{Chapter}%
  \def\sectionautorefname{Section}%
  \def\subsectionautorefname{Section}%
  \def\subsubsectionautorefname{Section}%
  \def\appendixautorefname{Section}%
  \def\definitionautorefname{Definition}%
    \def\theoremautorefname{Theorem}%
  \def\propositionautorefname{Proposition}%
  \def\corollaryautorefname{Corollary}%
    \def\algorithmautorefname{Algorithm}%
  \autoref{#1}%
  \endgroup%
}

\usepackage[sectionbib]{bibunits}
\defaultbibliographystyle{chicago}
\defaultbibliography{references.bib}

\title{\huge Efficient Bayesian inference for multiple network data}
\date{}
\author[1]{Francesco Barile}
\author[1]{Sara Capozio} 
\author[1]{Bernardo Nipoti}

\affil[1]{Department of Economics, Management and Statistics, University of Milano-Bicocca, Italy\\\href{mailto:francesco.barile@unimib.it}{\nolinkurl{francesco.barile@unimib.it}} 
}

\begin{document}

\maketitle

\begin{bibunit}

\begin{abstract}
We investigate distributional properties of the centered Erd\H{o}s--R\'enyi distribution \citep{Lun21} and propose a semi-conjugate Bayesian approach to multiple network data.
 In simulations, both Gibbs sampling and empirical Bayes accurately recover network summaries, with the latter scaling efficiently with network size. As a companion to this note, we provide the R package \texttt{BayesCER}, which implements the proposed methodology.
\end{abstract}

\noindent%


\section{Introduction}\label{sec1}
Multiple network data consist of repeated observations of networks on the same node set. They arise, for example, in neuroscience, where connections among brain regions are recorded across individuals or populations \citep{nelson2017comparison}. Their analysis becomes increasingly challenging as the number of nodes grows, since both the number of potential edges and the space of possible graph configurations grow rapidly. Much of the literature addresses this difficulty via dimensionality reduction, including latent-space models
\citep[e.g.,][]{gollini,durante2017,dangelo19,dangelo_fop}. The centered Erd\H{o}s--R\'enyi (CER) distribution of \citet{Lun21} offers a different strategy: a tractable location--dispersion model that retains the full graph representation without imposing specific structural assumptions on its central graph. The CER has also been used as a kernel in mixture models \citep{Man24,Bar24}.

This note makes two main contributions. First, for CER-distributed graphs, we derive exact and asymptotic edge-count distributions, expressions for expected triangle and length-two-path counts, and an approximation to the expected global clustering coefficient; node-specific counterparts are given in the appendices. Second, we introduce a generalization of CER as the prior for the location parameter in a semi-conjugate Bayesian model, yielding Gibbs and empirical Bayes (EB) procedures. Numerical experiments show that both recover global network summaries accurately, with EB scaling efficiently with network size.

\section{Modeling tools and properties}  
Let $\mathcal V=\{1,\ldots,N\}$ and let $\mathcal G_{\mathcal V}$ be the
space of simple undirected labeled graphs on $\mathcal V$. We also let $M=\binom{N}{2}$ be the maximum number of edges that a graph in $\mathcal{G}_{\mathcal{V}}$ may display. A graph $G$
is represented by the upper-triangular part of its adjacency matrix $A_G$. For notational convenience, we introduce $\bar{A}_G$, the symmetrized version of $A_G$, defined such that for any $i \leq j$, $\bar{A}_G[ij] = \bar{A}_G[ji] = A_G[ij]$. 
Finally, we use the notation $\tilde{G}$ for a random graph in this section.

\subsection{The centered Erdős--Rényi distribution}\label{sec:cer_sec}

\begin{definition}[\citealp{Lun21}]\label{def:CER}
A random graph $\tilde G$ taking values in $\mathcal G_{\mathcal V}$
has a centered Erdős--Rényi distribution with location
$C\in\mathcal G_{\mathcal V}$ and dispersion $\alpha\in(0,1/2)$,
written $\tilde G\sim\CER(C,\alpha)$, if the entries
$A_{\tilde G}[ij]$, $i<j$, are independent and
\begin{equation}\label{eq:CER1}
\Pr(A_{\tilde G}[ij]=A_C[ij])=1-\alpha,
\qquad i<j.
\end{equation}
\end{definition}
\noindent 
Its probability mass function (pmf) on $\mathcal{G}_{\mathcal{V}}$ can be written as $p_{\CER}(G;C,\alpha)=\alpha^{d_{\text{H}}(G,C)}(1-\alpha)^{M-d_{\text{H}}(G,C)}$, 
where $d_\text{H}$ denotes the Hamming distance, 
such that $d_{\text{H}}(G,C)$ counts the number of entries that differ between $A_G$ and $A_{C}$. 
\citet{Lun21} show that the CER distribution is unimodal with unique mode at $C$.
Moreover, 
the Fréchet mean set $\mathscr{F}(\tilde G)=\argmin_{G\in\mathcal{G}_{\mathcal{V}}}
\mathbb{E}[d_{\mathrm{H}}^{2}(\tilde G,G)]$, 
with 
expectation 
taken with respect to $\tilde G \sim \CER(C,\alpha)$, 
is a singleton whose unique element is $C$.
Finally, $\alpha$ regulates the dispersion around $C$ and directly controls the entropy of the distribution: 
the distribution concentrates at $C$ as $\alpha\to0$ and approaches the uniform distribution on $\mathcal G_{\mathcal V}$ as $\alpha\to1/2$.

We next derive additional properties of the CER distribution.
All proofs and node-specific results are provided in \ref{sec:proof_S1}. 
For a graph $G$, let $|G|=\sum_{i<j}A_G[ij]$ be its number of edges and let $|G|_{i}=\sum_{j\neq i}\bar{A}_G[ij]$ be the degree of node $i$.  Let $T(G)$ denote the number of triangles in $G$, namely the number of
triples $i<j<k$ such that
$\bar{A}_G[ij]=\bar{A}_G[ik]=\bar{A}_G[jk]=1$, and let
$P^{(2)}(G)=\sum_{i=1}^N\binom{|G|_i}{2}$ denote the number of
non-oriented paths of length two. For each node $i$,
$\binom{|G|_i}{2}$ counts the unordered pairs of its neighbors,
irrespective of whether they are connected. The following
results give the global edge-count distribution, and its large-$N$
approximation, and
the global expected number of triangles and non-oriented paths of length two.  In the following, the notation $\stackrel{\footnotesize{d}}{=}$ and $\stackrel{\footnotesize{d}}{\longrightarrow}$  indicates identity and convergence in distribution, respectively.
\begin{proposition}\label{prop:edges}
If $\tilde G \sim \CER(C,\alpha)$, then
$|\tilde G| \stackrel{d}{=} W_1 + W_2$,
    where $W_1 \sim \mathrm{Binomial}\left(|C|,1-\alpha\right)$ and $W_2 \sim \mathrm{Binomial}\left(M-|C|, \alpha\right)$ are independent.
In addition, as
$N\to\infty$,
\begin{equation}\label{eq:CLT_G}
\frac{|\tilde G| - \big(M\alpha + |C| (1-2\alpha)\big)}{\sqrt{M \alpha(1-\alpha)}} 
        \stackrel{d}{\longrightarrow} Z \sim \mathrm{N}(0,1).
\end{equation}
\end{proposition}
\noindent The expected edge density, that is the expected proportion of possible edges present in $\tilde G$, is a simple summary of overall connectivity.  
By \Autoref{prop:edges},
\begin{equation}\label{eq:edge-density}
\frac{\mathbb E[|\tilde G|]}{M}
=\alpha+(1-2\alpha)\frac{|C|}{M}.
\end{equation}

\begin{proposition}\label{prop:triangles}
If $\tilde G\sim\CER(C,\alpha)$, then the following hold.

\noindent\textup{(a)} The expected number of triangles in $\tilde G$ is
\begin{equation*}
\mathbb{E}[T(\tilde G)]
=
\alpha^3\binom{N}{3}
+\alpha^2(1-2\alpha)(N-2)|C|
+\frac{1}{2}\alpha(1-2\alpha)^2
\Bigg(\sum_{i=1}^N|C|_i^2-2|C|\Bigg)
+(1-2\alpha)^3T(C).
\end{equation*}

\noindent\textup{(b)} The expected number of non-oriented paths of length
two in $\tilde G$ is
\begin{equation*}
\mathbb{E}[P^{(2)}(\tilde G)]
=
3\alpha^2\binom{N}{3}
+(1-2\alpha)\{2(N-1)\alpha-1\}|C|
+\frac{1}{2}(1-2\alpha)^2\sum_{i=1}^N|C|_i^2.
\end{equation*}
\end{proposition}
\noindent Although not immediately interpretable, the expressions for
$\mathbb{E}[T(\tilde G)]$ and $\mathbb{E}[P^{(2)}(\tilde G)]$ in \Autoref{prop:triangles} are readily evaluated numerically, as
they depend on the location parameter $C$ only through its edge count
$|C|$, sum of squared degrees $\sum_{i=1}^N |C|_i^2$, and triangle count
$T(C)=\operatorname{tr}(\bar A_C^3)/6$. The random variables $T(\tilde G)$ and $P^{(2)}(\tilde G)$ are closely related to the tendency of nodes to form tightly connected clusters. One of the most commonly used measures of this tendency is the global
clustering coefficient, also known as transitivity \citep{Newman03},
defined for a network $G$ by
$\mathscr C(G)=3T(G)/P^{(2)}(G)$ when $P^{(2)}(G)>0$, and
$\mathscr C(G)=0$ otherwise. Following \citet{nakajima2021measurement}, we approximate the expected
global clustering coefficient as
\begin{equation}\label{eq:MF}
\mathbb{E}[\mathscr{C}(\tilde G)]
\approx
\frac{3\mathbb{E}[T(\tilde G)]}
{\mathbb{E}[P^{(2)}(\tilde G)]}.
\end{equation}
We call the right-hand side of \eqref{eq:MF} the approximate expected
global clustering coefficient and denote it by $\mathscr{C}^*(\tilde G)$. \Autoref{prop:gcc} gives it explicitly for
$\tilde G\sim\CER(C,\alpha)$.
\begin{corollary}
\label{prop:gcc} If $\tilde G\sim\CER(C,\alpha)$, then the approximate expected global clustering coefficient $\mathscr{C}^*(\tilde G)$ is given by
\begin{equation*}
\mathscr{C}^*(\tilde G)=\frac{
        3\alpha^3 \binom{N}{3}
        + 3\alpha^2(1-2\alpha)(N-2)|C|
        + \frac{3}{2}
        \alpha(1-2\alpha)^2 \left(\sum_{i=1}^N |C|_i^2 - 
        2|C| \right)
        + 3 (1-2\alpha)^3 T(C)
}{
        3\alpha^2 \binom{N}{3}
        + (1-2\alpha)(2(N-1)\alpha-1)|C|
        + \frac{1}{2} (1-2\alpha)^2 \sum_{i=1}^N |C|_i^2}. \label{eq:ecc}
\end{equation*}
\end{corollary}
\subsection{A generalized CER distribution}
Instrumental to the construction of a conjugate prior for $C$, we
propose an extension of the support of the CER location parameter. 
Let
$\mathcal W_{\mathcal V}\supset\mathcal G_{\mathcal V}$ be the space of simple weighted
 undirected graphs
 whose upper-triangular adjacency entries 
 $A_{W}[ij] \in [0,1], i<j$, where $W \in \mathcal W_{\mathcal V}$. 
For $W_1,W_2\in\mathcal W_{\mathcal V}$, let 
$d_\mathrm{T}(W_1,W_2)=\sum_{i<j}|A_{W_1}[ij]-A_{W_2}[ij]|$ denote the taxicab, or $L^1$, distance between $W_1$ and $W_2$, and observe that, when $W_1,W_2\in\mathcal{G}_{\mathcal{V}}$, then $d_\mathrm{T}(W_1,W_2)=d_\mathrm{H}(W_1,W_2)$. 
As above, we let $\bar{A}_W$ denote the symmetrization of $A_W$.
\begin{definition}
   A random graph $\tilde G$ taking values in $\mathcal{G_{\mathcal{V}}}$ has generalized centered Erdős--Rényi (GCER) distribution with location $W\in\mathcal{W}_{\mathcal{V}}$ and 
   dispersion $\alpha\in(0,1/2)$, if the $A_{\tilde{G}}[ij]$'s are independent and
\begin{equation}\label{eq:CER2}
\Pr(A_{\tilde G}[ij]=1
)=\frac{1}{1+\pi_{ij}(W,\alpha)}, \quad \text{ where  }\quad \pi_{ij}(W,\alpha)=\left(\frac{\alpha}{1-\alpha}\right)^{2 A
_W[ij]-1}.
\end{equation} 
We write $\tilde G\sim\GCER(W,\alpha)$ and
$p_{\GCER}(G;W,\alpha)$ for the corresponding pmf.
\end{definition}
    \noindent The CER distribution is recovered when $W \in \mathcal{G}_{\mathcal{V}}$.
    In a CER-distributed random graph $\tilde G$, the probability of an edge in \eqref{eq:CER1}
    is restricted to be either $\alpha$ or $1-\alpha$.
  The GCER distribution relaxes this constraint by allowing the corresponding probability in \eqref{eq:CER2} to vary over $[\alpha,1-\alpha]$. 
    Propositions~\ref{prop:GCERmonotone} and~\ref{prop:GCER_mode} show that $W$ acts as a location parameter. 
    \Autoref{prop:entropy} shows that, for any fixed $W\in\mathcal{W}_\mathcal{V}$ such that \(A_W[ij]\neq 1/2\) for at least one pair \(i<j\), the entropy of a random graph \(\tilde G\sim\GCER(W,\alpha)\), that is $-\mathbb{E} \left[ \log p_{\GCER}(\tilde G;W,\alpha) \right]$, 
is a strictly increasing function of \(\alpha\), supporting its interpretation as a dispersion parameter. All proofs are provided in \ref{sec:proof_S2}.

\begin{proposition}\label{prop:GCERmonotone} 
Let \(W\in\mathcal{W}_{\mathcal{V}}\) and \(\alpha\in(0,1/2)\). If $G_1,G_2\in\mathcal{G}_{\mathcal{V}}$ are such that
$d_{\mathrm{T}}(G_1, W) \geq d_{\mathrm{T}}(G_2, W)$, then $p_{\GCER}(G_2;W,\alpha) \geq p_{\GCER}(G_1;W,\alpha)$, with 
equality holding if and only if $d_{\mathrm{T}}(G_1, W) = d_{\mathrm{T}}(G_2, W)$. 
\end{proposition}
\begin{proposition}\label{prop:GCER_mode}
Let $\tilde G \sim \GCER(W, \alpha)$. 
The Fr\'echet mean set $\mathscr{F}(\tilde G)$ of $\tilde G$, under the Hamming distance, is given by
\begin{equation*}
    \mathscr{F}(\tilde G)=\left\{L\in\mathcal{G}_{\mathcal{V}}\,:\,A_L[ij]=\mathbbm{1}_{\{A_W[ij]>1/2\}}
    \text{ for any }i<j \text{ s.t. } A_W[ij]\neq 1/2\right\}.
\end{equation*}
Moreover, the set of modes of the distribution of $\tilde G$ coincides with $\mathscr{F}(\tilde G)$.
\end{proposition}
\noindent Each mode of the distribution $\tilde G \sim \GCER(W, \alpha)$ thus corresponds to a discretization of the entries of $A_W$. 
The number of modes is $|\mathscr{F}(\tilde G)|=2^{M^*}$, 
where $M^*=\sum_{i<j}\mathbbm{1}_{\{A_W[ij]=1/2\}} \in\{0,1,\ldots,M\}$, 
since only the entries satisfying \(A_W[ij]=1/2\) admit two possible
discretizations, while all the remaining entries are uniquely determined by their position
relative to \(1/2\).

\begin{proposition}\label{prop:entropy}
If $\tilde G\sim\GCER(W,\alpha)$, then, for fixed
$W\in\mathcal W_{\mathcal V}$, its entropy $H(\tilde G)$ is
nondecreasing in $\alpha\in(0,1/2)$ and strictly increasing if
$A_W[ij]\neq1/2$ for some $i<j$. If $M^*<M$, then
\textup{(i)} $H(\tilde G)\to M^*\log(2)$ as $\alpha\to0$ and
\textup{(ii)} $H(\tilde G)\to M\log(2)$ as $\alpha\to1/2$. If
$M^*=M$, then $H(\tilde G)=M\log(2)$ for all
$\alpha\in(0,1/2)$.
\end{proposition}

\noindent The entropy $H(\tilde G)$ provides a convenient measure of dispersion for
$\tilde G\sim\GCER(W,\alpha)$. By 
\Autoref{prop:entropy}, in the
unimodal case $M^*=0$, the entropy of $\tilde{G}$  ranges over $(0,M\log(2))$, with the lower and
upper bounds approached only as $\alpha\to0$ and $\alpha\to1/2$,
respectively. At the opposite extreme, the case \(M^*=M\) corresponds to
\(A_W[ij]=1/2\) for all \(i<j\). In this case, the \(\GCER(W,\alpha)\) distribution reduces to an
Erdős--Rényi distribution with edge probability \(1/2\). Consequently, the
entropy of \(\tilde G\) attains its largest possible value, \(M\log (2)\), for every
\(\alpha\).
\begin{remark}
The parameters $(W,\alpha)$ are not jointly identifiable under the GCER
model: distinct parameter pairs may induce the same edge probabilities
$\Pr(A_{\tilde G}[ij]=1)$ for all $i<j$, and hence the same pmf
on $\mathcal G_{\mathcal V}$. Nevertheless, distributional quantities
such as the Fr\'echet mean set and entropy depend only on the induced
edge probabilities, not on the particular parametrization. Any two parametrizations inducing the same distribution
yield the same discretization of their location matrices at threshold $1/2$ and hence the same Fr\'echet mean set. For fixed $W$ with $M^*<M$, $\alpha$ continues to control entropy and thus dispersion.
\end{remark}

\section{Bayesian modeling and posterior computation}\label{sec:model}
We propose the following model for a vector $G^{(1:n)}=(G_1,\ldots,G_n)$ of observations on $\mathcal{G}_{\mathcal{V}}$:
\begin{equation}\label{eq:model}
\begin{aligned}
    G_\ell\mid C,\alpha&\simiid \CER(C,\alpha)& \ell=1,\ldots,n\\ 
   C\mid \alpha &\sim \GCER(W,\tau_{n_0}(\alpha))&\\
    \alpha& \sim \mathrm{TBeta}(1/2,a,b),&
    \end{aligned}
\end{equation}
where, for $a,b>0$, $\mathrm{TBeta}(1/2,a,b)$ denotes the truncated Beta distribution on $(0,1/2)$ and
the scaling function $\tau_{n_0}:(0,1/2)\rightarrow (0,1/2)$, illustrated in 
\Autoref{fig:scaling} 
in \ref{app:app_scaling}, 
is defined as
\begin{equation*}
    \tau_{n_0}(\alpha)=\frac{\alpha^{n_0}}{(1-\alpha)^{n_0}+\alpha^{n_0}}, \qquad n_0 >0.
\end{equation*}
\noindent 
The map $\tau_{n_0}$ allows the prior and sampling dispersions to differ
while preserving conjugacy for $C$, with $n_0$ controlling their relationship. Our specification is advantageous in that it yields an analytically
tractable full conditional distribution for $C$ and, under the condition
stated below, also for $\alpha$, as formalized in the next two propositions.
\begin{proposition}\label{prop:fc_Gm}
    Under model \eqref{eq:model}, 
    $C\mid G^{(1:n)},\alpha \sim \GCER(W^\prime,\tau_{n_0+n}(\alpha))$, where
    \begin{equation}\label{eq:post_GCER}
        A_{W^\prime}=\frac{n_0}{n_0+n}A_{W}+\frac{n}{n_0+n}A_{\overline{G}},
    \end{equation}
    and $A_{\overline{G}}=n^{-1}\sum_{\ell=1}^n A_{G_\ell}$ is the edge-by-edge average of the adjacency matrices of $G^{(1:n)}$. Hence, 
    the GCER distribution is a conjugate prior for the location parameter $C$ of the CER model.
\end{proposition}
\noindent The location matrix $A_{W'}$ in the full conditional distribution of
$C$ is a convex combination of the prior location matrix $A_W$ and the
empirical mean $A_{\bar G}$, with weights proportional to $n_0$ and $n$,
respectively. By
\Autoref{prop:GCER_mode}, $W'$ determines the Fr\'echet mean set
of the full conditional distribution. Thus, in general, $n_0$ controls
the influence of the prior relative to the data in determining the
posterior location.

\begin{proposition}\label{prop:fc_alpha}
Under model \eqref{eq:model}, let
$D(C)=\sum_{\ell=1}^n d_{\mathrm H}(G_\ell,C)$. If $W=W_0$, with
$A_{W_0}[ij]=1/2$ for all $i<j$, then
\begin{equation*}
\alpha\mid G^{(1:n)},C
\sim\mathrm{TBeta}(1/2,a+D(C),b+nM-D(C)).
\end{equation*}
\end{proposition}
\noindent Note that if model \eqref{eq:model} is specified with $W=W_0$, the GCER
prior reduces, regardless of $n_0$, to an Erd\H{o}s--R\'{e}nyi prior
with edge probability $1/2$, a natural choice to reflect prior ignorance. In this case,
the full conditional distributions of both $C$ and $\alpha$ belong to
the same parametric families as their respective priors. 
We refer to the specification of model \eqref{eq:model} with $W=W_0$ as
the \emph{semi-conjugate} specification.

We consider two computational strategies, both implemented in the \texttt{R} package \texttt{BayesCER} publicly available at \url{https://github.com/FrancescoBarile/BayesCER}. The first is posterior
sampling, implemented under the semi-conjugate specification by a simple
Gibbs sampler requiring only direct draws from the full conditional
distributions in Propositions~\ref{prop:fc_Gm}
and~\ref{prop:fc_alpha}. For general \(W\), 
\Autoref{prop:fc_Gm} still gives the update for
\(C\), while the nonstandard full conditional of \(\alpha\) may be sampled
by a one-dimensional Metropolis--Hastings step. The second strategy is an EB procedure that avoids
posterior sampling. Let
$s_{ij}=\sum_{\ell=1}^n A_{G_\ell}[ij]$ and
$r_{ij}=n_0(1-2A_W[ij])$. Integrating out $C$ yields the marginal
 likelihood
\begin{equation}\label{eq:marginal-likelihood}
L(\alpha;G^{(1:n)})=
\prod_{i<j}
\frac{
\alpha^{s_{ij}}(1-\alpha)^{n-s_{ij}+r_{ij}}
+\alpha^{n-s_{ij}+r_{ij}}(1-\alpha)^{s_{ij}}
}{
\alpha^{r_{ij}}+(1-\alpha)^{r_{ij}}
}.
\end{equation}
Under the semi-conjugate specification, $r_{ij}=0$ for all $i<j$, so
\eqref{eq:marginal-likelihood} reduces to
\eqref{eq:mar_lik_conj}
in \ref{app:app_ml}.

Assuming that the marginal likelihood attains its maximum in
\((0,1/2)\), we choose
\[\hat{\alpha}\in\argmax_{\alpha\in(0,1/2)}
L(\alpha;G^{(1:n)}).\] Under the semi-conjugate specification, this
maximizer is unique provided suitable data conditions hold; see \ref{sec:unique}. 

Plugging $\hat{\alpha}$ into 
\Autoref{prop:fc_Gm} gives
$C\mid G^{(1:n)},\alpha=\hat{\alpha}
\sim\GCER(W',\tau_{n_0+n}(\hat{\alpha}))$, where the adjacency matrix of $W'$ is given in
\eqref{eq:post_GCER}. Under a $0$--$1$ loss, any mode of this
plug-in posterior is an EB estimator of $C$. One such mode, henceforth
denoted by $\hat C$, is obtained by thresholding $A_{W'}$ at $1/2$,
that is, by setting
$A_{\hat C}[ij]=\mathbbm{1}_{\{A_{W'}[ij]\geq1/2\}}$ for each $i<j$.

The expected edge density and approximate expected global clustering
coefficient derived in 
\Autoref{sec:cer_sec} provide population-level
summaries of network connectivity and clustering. Under Gibbs sampling,
posterior inference for these summaries is obtained by evaluating the corresponding expressions at each draw of $(C,\alpha)$; under EB, they
are estimated by plugging in $(\hat C,\hat\alpha)$. Derivations of the full conditionals in
Propositions~\ref{prop:fc_Gm} and~\ref{prop:fc_alpha} and the marginal
likelihood in \eqref{eq:marginal-likelihood} are given in \ref{sec:sm_sec3}. 

\section{Numerical experiments}\label{sec:application}
We conduct a simple simulation experiment to assess whether Gibbs and EB
inference remain computationally feasible for reasonably large $n$ and $N$
and whether posterior and plug-in estimates accurately recover
population-level network functionals. Specifically, we consider
$n\in\{10,100,1000\}$, $N\in\{10,100,1000,2000\}$, and $\alpha=0.2$.
For each $N$, the location graph $C$ was drawn from a two-community stochastic
block model with within-community edge probability $0.8$ and
between-community probability $0.1$, with nodes assigned to the two communities independently and with equal probability; conditional on $C$, the observations
were independently generated from $\CER(C,\alpha)$. For each of the
resulting $12$ settings, we generated $100$ independent datasets.
We adopted the semi-conjugate specification $W=W_0$, under which the induced prior and posterior distributions do not depend on $n_0$, and set $a=b=1$. We ran the Gibbs sampler for 11,000 iterations, discarding the
first 1,000 as burn-in and retaining every tenth subsequent draw, yielding
1,000 posterior draws. 

\begin{figure}[h]
\centering
\includegraphics[width=\linewidth,clip,
 trim={0 20mm 0 17mm}]{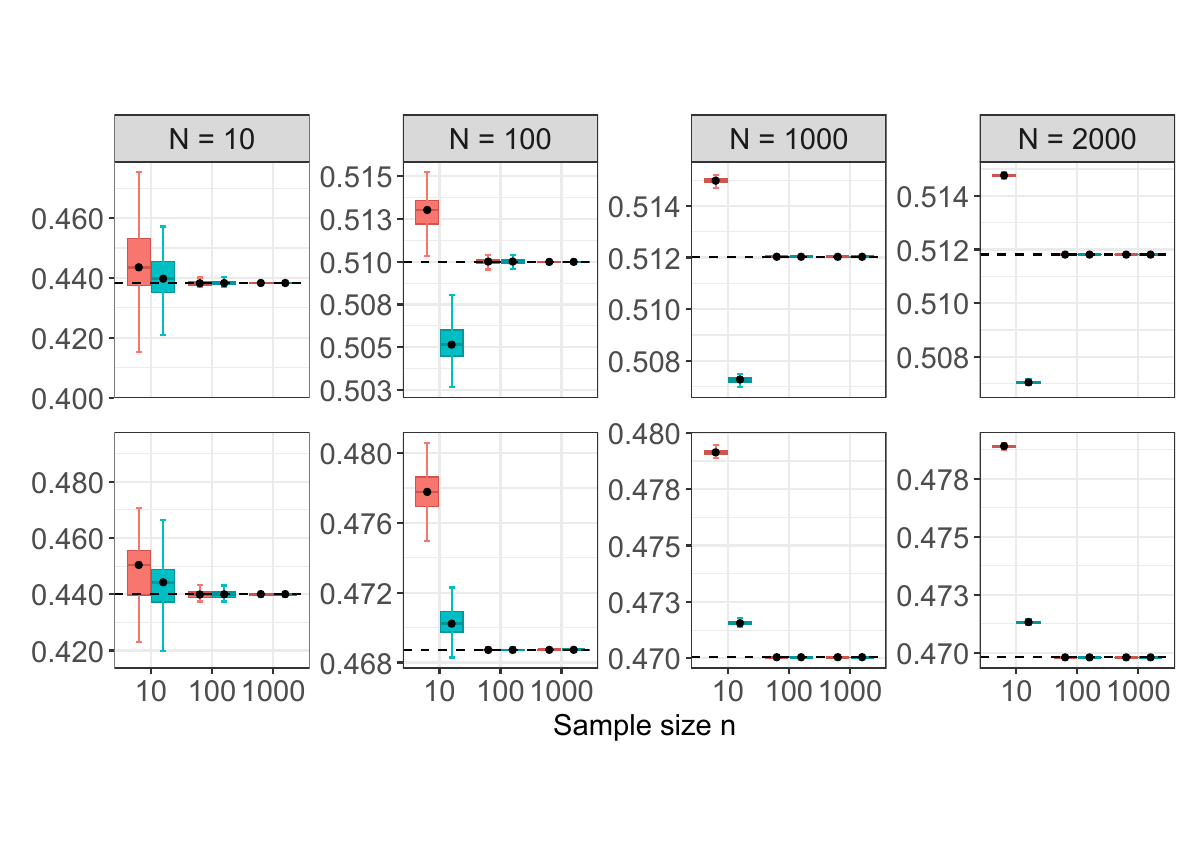}
\caption{Approximate expected global clustering coefficient (top row) and expected
edge density (bottom row). Coral and green boxplots correspond to EB and
Gibbs, respectively, and summarize 100 replications for each combination
of $N$ and $n$. Dashed black lines indicate the data-generating values;
black dots mark the median estimates.}
\label{fig:unica}
\end{figure}

\Autoref{fig:unica} reports estimates of the
approximate expected global clustering coefficient derived in
\Autoref{prop:gcc} and the expected edge density in
\eqref{eq:edge-density}. 
For Gibbs, each estimate is the posterior average of the corresponding
functional over the retained draws, whereas EB uses its plug-in value at
$(\hat C,\hat\alpha)$. At $n=10$, the two procedures exhibit distinct
finite-sample behavior: EB tends to overestimate both summaries, while
Gibbs underestimates the clustering coefficient for $N\geq100$ and is
closer to the data-generating edge density. The greatest variability occurs
for $N=n=10$. These differences largely disappear as $n$ increases: for
$n\geq100$, the estimates from both methods are tightly concentrated around
the data-generating values and are nearly indistinguishable across all
values of $N$. The main difference is computational: for example, at $N=2000$ and
$n=1000$, the median running times, including the computation of the network summaries,
were $4.11$ seconds for EB and
$3\,424.13$ seconds for Gibbs,
on a machine with an Intel Xeon Platinum 8370C processor (16 logical cores) and 32 GB of RAM, with additional timing results reported in \Autoref{tab:computational_time} in \ref{sec:sm_sec4}. 
Both procedures complete the largest setting considered, but EB provides
comparable point estimates at a small fraction of the computational cost. The results above concern the representative case $\alpha=0.2$;
corresponding global summaries for $\alpha\in\{0.1,0.3\}$, together with
the normalized Hamming error for $C$ and the absolute error for $\alpha$,
are reported in the Supplementary Material
and support the same conclusions.

\paragraph{AI use declaration} During the preparation of this work, the authors used ChatGPT (GPT-5.6, OpenAI) to identify potential errors and improve the mathematical exposition; they critically reviewed and selectively incorporated its suggestions and take full responsibility for the article’s content.

\putbib
\end{bibunit}

\clearpage

\begin{bibunit}

\bigskip

\allowdisplaybreaks

\begin{center}
{\huge
\bf Appendices}
\end{center}

These appendices collect the node-specific results, all proofs, details on
Bayesian computation, additional numerical experiments, and computational
timing results. Recall that, for a graph \(G\), \(|G|_i\) denotes the degree
of node \(i\). Moreover, let $T_i(G)$ denote the number of triangles containing node $i$, and let
$P_i^{(2)}(G)=\binom{|G|_i}{2}$ 
denote the number of non-oriented paths of length two with $i$ as their middle node.

\appendix 

\makeatletter
\renewcommand{\thesection}{Appendix \Alph{section}}
\makeatother

\section{CER distribution: node-specific results and proofs}\label{sec:proof_S1}
The exact and asymptotic distributions of the degree of a fixed node are given in the following proposition.
\begin{proposition}
\label{prop:edges_i}
Let $\tilde G \sim \CER(C,\alpha)$. 
Then, for any node $i \in \mathcal{V}$, the degree of the $i$th node in $\tilde G$ satisfies
    $|\tilde G|_{i} \stackrel{d}{=} W_{1i} + W_{2i}$, 
    where $W_{1i} \sim \mathrm{Binomial}\left(|C|_i, 1-\alpha\right)$ and $W_{2i} \sim \mathrm{Binomial}\left(N-1-|C|_i, \alpha\right)$ are independent. In addition, as
$N\to\infty$,
 \begin{equation}\label{eq:CLT_G_i}
        \frac{|\tilde G|_{i} - \big((N-1)\alpha + |C|_i (1-2\alpha)\big)}{\sqrt{(N-1)\alpha(1-\alpha)}} 
        \stackrel{d}{\longrightarrow} Z \sim \mathrm{N}(0,1).
    \end{equation}
\end{proposition}

\begin{proof}[Proof of 
\Autoref{prop:edges_i}] The first result follows immediately from the fact that the elements of $A_{\tilde G}$ are independent Bernoulli random variables. Indeed, following 
\Autoref{def:CER},
we have that, for any $i<j\in \mathcal{V}$, $A_{\tilde G}[ij]\sim \text{Bern}(\alpha)$ if $A_C[ij]=0$, while $A_{\tilde G}[ij]\sim \text{Bern}(1-\alpha)$ if $A_C[ij]=1$.
For the asymptotic result, for any $i \in \mathcal{V}$, we consider $N-1$ independent Bernoulli random variables $X_1, \ldots, X_{N-1}$ defined as
\begin{equation*}
    X_j=\begin{cases}
        \bar{A}_{\tilde G}[ij]& \text{if }j=1,\ldots,i-1;\\
        \bar{A}_{\tilde G}[i(j+1)]& \text{if }j=i,\ldots,N-1.
    \end{cases}
\end{equation*}
Each $X_j$ takes values in $\{0,1\}$ and has finite mean $\mu_j$, equal to either $\alpha$ or $1-\alpha$, 
and finite variance $\sigma_j^2 = \alpha(1-\alpha)$. A direct application of Lyapunov central limit theorem \citep[see, e.g.,][]{Bil17} shows that the result in \eqref{eq:CLT_G_i} holds if  \begin{equation}\label{eq:lya}
\lim_{N \to \infty} \frac{1}{s_{N-1}^{2+\delta}} \sum_{j=1}^{N-1} \mathbb{E} \left[ |X_j - \mu_j|^{2+\delta} \right] = 0, \text{ for some }\delta>0,
\end{equation}
where \( s_{N-1}^2 = \sum_{j=1}^{N-1} \sigma_j^2=(N-1)\alpha(1-\alpha)\).
We observe that, for every $\delta>0$,
\begin{align*}
\mathbb{E} \left[ |X_j - \mu_j|^{2+\delta} \right] 
&= (1 - \mu_j) \mu_j^{2+\delta} + \mu_{j} (1 - \mu_{j})^{2+\delta} \\
&= \mu_{j} (1 - \mu_{j}) \left( \mu_{j}^{1+\delta} + (1 - \mu_{j})^{1+\delta} \right) \\
&\leq \mu_{j} (1 - \mu_{j})=\alpha(1-\alpha).
\end{align*}
Thus, for any \(\delta > 0\),
\begin{align}\label{eq:lya2}
\frac{1}{s_{N-1}^{2+\delta}} \sum_{j=1}^{N-1} \mathbb{E} \left[ |X_j - \mu_j|^{2+\delta} \right] 
&\leq \frac{\sum_{j=1}^{N-1}\alpha(1 - \alpha)}{((N-1)\alpha(1 - \alpha))^{\frac{\delta}{2}+1}} = \frac{1}{((N-1)\alpha(1 - \alpha))^{\frac{\delta}{2}}}.
\end{align}
Since the right-hand side of \eqref{eq:lya2} converges to zero as $N \to \infty$, 
the left-hand side, being non-negative by definition, must also converge to zero. 
Therefore, the Lyapunov condition \eqref{eq:lya} is satisfied, and \eqref{eq:CLT_G_i} follows.
\end{proof}

The proof of \Autoref{prop:edges} follows directly from \Autoref{prop:edges_i}, as shown below.

\begin{proof}[Proof of 
\Autoref{prop:edges}]
The argument follows exactly the same lines as in the previous proof, with the only difference that the $N-1$ independent Bernoulli random variables are replaced by the $M=\binom{N}{2}$ independent Bernoulli variables given by the collection $\{A_{\tilde G}[ij] : i<j,\; i,j\in\mathcal{V}\}$.
\end{proof}

The expected numbers of triangles and non-oriented paths of length two involving a fixed node are given in the following proposition.

\begin{proposition}\label{prop:p2}
Let $\tilde G \sim \CER(C,\alpha)$.
Then, for any node $i \in \mathcal{V}$:
\begin{itemize}
   \item[(a)] the expected number of triangles in $\tilde G$ containing node $i$  is 
\begin{align*}
    \mathbb{E}[T_i(\tilde G)]=&\alpha^3\binom{N-1}{2}+\alpha^2(1-2\alpha)\left((N-3)|C|_i+|C|\right)
    \\& 
    +\alpha\left(1-2\alpha\right)^2 \left(
     \frac{|C|_i\left(|C|_i-3\right)}{2} + \sum_{j \neq i} \bar{A}_C[ij]|C|_j \right)
    +\left(1-2\alpha\right)^3T_i(C).
\end{align*}
   \item[(b)] the expected number of non-oriented paths of length two in $\tilde G$ having $i$ as their middle node is 
    \begin{equation*}
        \mathbb{E}[P_{i}^{(2)}(\tilde G)] = 
         \alpha^{2}
         \binom{N-1}{2} + 
         \frac{1}{2}
         (1-2\alpha) |C|_i \Bigl( (1-2\alpha) |C|_i + 2\alpha(N-1) - 1 \Bigl).
    \end{equation*} 
\end{itemize}
\end{proposition}

\begin{proof}[Proof of \Autoref{prop:p2}]
We prove parts (a) and (b) separately. 
\begin{itemize}
    \item[(a)] The random variable \(T_i(\tilde G)\) can be written as $$T_i(\tilde G)=\sum_{\mathcal{J}_i}
\bar{A}_{\tilde G}[ij]\bar{A}_{\tilde G}[ik]\bar{A}_{\tilde G}[jk],$$
where $\mathcal{J}_i=\{(j,k)\,:\,j<k, j\neq i, k\neq i\}$. Thus, the expected value of \(T_i\) is given by
\begin{align*}
    \mathbb{E}[T_i(\tilde G)] = \sum_{\mathcal{J}_i} \mathrm{Pr}(\bar{A}_{\tilde G}[ij]=1, \bar{A}_{\tilde G}[jk]=1, \bar{A}_{\tilde G}[ik]=1).
\end{align*}
Since $\tilde G\sim\CER(C,\alpha)$, we obtain:
\begin{equation}\label{eq:ETi}
\mathbb{E}[T_i(\tilde G)] = \alpha^{3} \sum_{\mathcal{J}_i} \left(\frac{1-\alpha}{\alpha}\right)^{\bar{A}_C[ij]+\bar{A}_C[ik]+\bar{A}_C[jk]}.
\end{equation}
We observe that, for any pair $i<j$, 
\begin{equation*}
    \left(\frac{1-\alpha}{\alpha}\right)^{\bar{A}_C[ij]}=1+\frac{1-2\alpha}{\alpha}\bar{A}_C[ij], 
\end{equation*}
which implies that for any triplet of distinct nodes $\{i,j,k\}$ such that $j<k$,
\begin{align*}
    \left(\frac{1-\alpha}{\alpha}\right)^{\bar{A}_C[ij]+\bar{A}_C[ik]+\bar{A}_C[jk]}=&\left(1+\frac{1-2\alpha}{\alpha}\bar{A}_C[ij]\right)\left(1+\frac{1-2\alpha}{\alpha}\bar{A}_C[ik]\right)\left(1+\frac{1-2\alpha}{\alpha}\bar{A}_C[jk]\right).
\end{align*}
By expanding the product on the right-hand side of the previous equation and taking the sum of both sides over $\mathcal{J}_i$, we get
\begin{align}\label{eq:sum_expanded1i}
\sum_{\mathcal{J}_i}&\left(\frac{1-\alpha}{\alpha}\right)^{\bar{A}_C[ij]+\bar{A}_C[ik]+\bar{A}_C[jk]}=\sum_{\mathcal{J}_i}1\\\label{eq:sum_expanded2i}
    &\qquad+\frac{1-2\alpha}{\alpha}\sum_{\mathcal{J}_i}(\bar{A}_C[ij]+\bar{A}_C[ik]+\bar{A}_C[jk])\\\label{eq:sum_expanded3i}
    &\qquad+\left(\frac{1-2\alpha}{\alpha}\right)^2\sum_{\mathcal{J}_i}(\bar{A}_C[ij]\bar{A}_C[ik]+\bar{A}_C[ij]\bar{A}_C[jk]+\bar{A}_C[ik]\bar{A}_C[jk])\\\label{eq:sum_expanded4i}
    &\qquad+\left(\frac{1-2\alpha}{\alpha}\right)^3 \sum_{\mathcal{J}_i}\bar{A}_C[ij]\bar{A}_C[ik]\bar{A}_C[jk].
\end{align}
We further observe that:
\begin{itemize}
    \item[(a1)] The number of distinct pairs $\{j,k\}$ such that $j,k\neq i$ is $\binom{N-1}{2}$, and thus the term in \eqref{eq:sum_expanded1i} gives 
    \begin{align*}
 \sum_{\mathcal{J}_i}1&=\binom{N-1}{2}.
 \end{align*}
 \item[(a2)] 
The terms $\bar{A}_C[ij]$ and $\bar{A}_C[ik]$ each appear exactly $N-2$ times in the sum over $\mathcal{J}_i$, 
because for each neighbor $j$ of node $i$, there are $N-2$ possible nodes $k \neq i,j$. 
Hence, the sum over edges incident to $i$ is
\[
\sum_{\mathcal{J}_i} (\bar{A}_C[ij] + \bar{A}_C[ik]) = (N-2) |C|_i.
\]
Moreover, \(\sum_{\mathcal{J}_i} \bar{A}_C[jk]\) counts the edges in \(C\) that do not involve node \(i\). That is,
\begin{equation*}
    \sum_{\mathcal{J}_i} \bar{A}_C[jk] = |C| - |C|_i.
\end{equation*}
In summary, the sum in \eqref{eq:sum_expanded2i} simplifies to
\[
\sum_{\mathcal{J}_i} (\bar{A}_C[ij] + \bar{A}_C[ik] + \bar{A}_C[jk])=(N-3)|C|_i + |C|,
\]

\item[(a3)]
The term $\sum_{\mathcal{J}_i} \bar{A}_C[ij]\bar{A}_C[ik]$ counts all wedges centered at node $i$, which is simply
\[
\sum_{\mathcal{J}_i} \bar{A}_C[ij]\bar{A}_C[ik] = \binom{|C|_i}{2}.
\]

Moreover,  $\sum_{\mathcal{J}_i} (\bar{A}_C[ij]\bar{A}_C[jk] + \bar{A}_C[ik]\bar{A}_C[jk])$ counts wedges where $i$ is adjacent to one vertex and the edge $(j,k)$ forms the base. This sum can be written as
\[
\sum_{\mathcal{J}_i} (\bar{A}_C[ij]\bar{A}_C[jk] + \bar{A}_C[ik]\bar{A}_C[jk]) 
= 
\sum_{j \neq i} \bar{A}_C[ij] \, (|C|_j - 1).
\]
As a result, the sum in \eqref{eq:sum_expanded3i} simplifies to
\begin{align*}
\sum_{\mathcal{J}_i} \bigl( \bar{A}_C[ij]\bar{A}_C[ik] + \bar{A}_C[ij]\bar{A}_C[jk] + \bar{A}_C[ik]\bar{A}_C[jk] \bigr)
&= \binom{|C|_i}{2} + 
\sum_{j \neq i} \bar{A}_C[ij] \, (|C|_j - 1)\\
&= \frac{|C|_i\left(|C|_i-3\right)}{2} + \sum_{j \neq i} \bar{A}_C[ij]|C|_j.
\end{align*}
\item[(a4)] $\sum_{\mathcal{J}_i}\bar{A}_C[ij]\bar{A}_C[ik]\bar{A}_C[jk]$ counts the number of triangles in $C$ involving node $i$, hence the sum in \eqref{eq:sum_expanded4i} can be written as
\begin{equation*}
\sum_{\mathcal{J}_i}\bar{A}_C[ij]\bar{A}_C[ik]\bar{A}_C[jk]=T_i(C).
\end{equation*}
\end{itemize}
By replacing the results in (a1)--(a4) into \eqref{eq:sum_expanded1i}--\eqref{eq:sum_expanded4i}, and by multiplying by $\alpha^3$, the equation in \eqref{eq:ETi} becomes
\begin{align*}
    \mathbb{E}[T_i(\tilde G)]=&\alpha^3\binom{N-1}{2}+\alpha^2(1-2\alpha)\left((N-3)|C|_i+|C|\right)\\
     &+\alpha\left(1-2\alpha\right)^2 \left(
     \frac{|C|_i\left(|C|_i-3\right)}{2} + \sum_{j \neq i} \bar{A}_C[ij]|C|_j \right)
    +\left(1-2\alpha\right)^3T_i(C).
\end{align*}
\item[(b)]
Recall that $P_i^{(2)}(G)=\binom{|G|_i}{2}$. Since $|\tilde G|_{i}$ is the sum of two binomial variables with expected value:
\[
\mathbb{E}[|\tilde G|_{i}] = \alpha (N-1) + (1-2\alpha) |C|_{i},
\]
and variance:
\[
\operatorname{Var}(|\tilde G|_{i}) = \alpha (1-\alpha) (N-1),
\]
we can express the expectation of \( |\tilde G|_{i}(|\tilde G|_{i}-1) \) as
\[
\mathbb{E}[|\tilde G|_{i}(|\tilde G|_{i}-1)] = \mathbb{E}[|\tilde G|_{i}^2] - \mathbb{E}[|\tilde G|_{i}] = \operatorname{Var}(|\tilde G|_{i}) + \mathbb{E}[|\tilde G|_{i}](\mathbb{E}[|\tilde G|_{i}] - 1).
\]
Expanding this, we obtain
\[
\mathbb{E}[|\tilde G|_{i}(|\tilde G|_{i}-1)] = \alpha^{2}(N-1)(N-2) + (1-2\alpha) |C|_{i} \left[ (1-2\alpha) |C|_{i} + 2\alpha(N-1) - 1 \right].
\]
Multiplying by $\frac{1}{2}$, we obtain
 \begin{equation*}
        \mathbb{E}[P_{i}^{(2)}(\tilde G)] = 
         \alpha^{2}
         \binom{N-1}{2} + 
         \frac{1}{2}
         (1-2\alpha) |C|_i \Bigl( (1-2\alpha) |C|_i + 2\alpha(N-1) - 1 \Bigl).
\end{equation*} 
\end{itemize}
\end{proof}

\Autoref{prop:p2} also provides the starting point for the proof of \Autoref{prop:triangles}, as shown below.

\begin{proof}[Proof of \Autoref{prop:triangles}]
We prove parts (a) and (b) separately. 
\begin{itemize}
    \item[(a)]
The global and local triangle counts are related by the identity
\begin{align}\label{eq:T_Ti}
    T(G)= \frac{1}{3}\sum_{i=1}^{N} T_i(G).
\end{align}
By combining \eqref{eq:T_Ti} and \eqref{eq:ETi}, we get
\begin{align}\notag
\mathbb{E}[T(\tilde G)]&=\frac{\alpha^3}{3}\sum_{i=1}^N\binom{N-1}{2}+\frac{\alpha^2(1-2\alpha)}{3}\sum_{i=1}^N\left((N-3)|C|_i+|C|\right)\\\label{eq:ET_sums}
          &+\frac{\alpha\left(1-2\alpha\right)^2}{3} \sum_{i=1}^N\left(\frac{|C|_i(|C|_i-3)}{2}+\sum_{j\neq i}\bar{A}_C[ij]|C|_j\right)
    +\frac{\left(1-2\alpha\right)^3}{3}\sum_{i=1}^N T_i(C).
\end{align}  
We observe that:
\begin{itemize}
    \item[(a1)] It is immediate to check that 
    \begin{equation*}
        \sum_{i=1}^N\binom{N-1}{2}=3\binom{N}{3}.
    \end{equation*}
 \item[(a2)] Given that $\sum_{i=1}^N(N-3)|C|_i=2(N-3)|C|$ and $\sum_{i=1}^N |C|=N|C|$, we conclude that
 \begin{align*}
 \sum_{i=1}^N((N-3)|C|_i+|C|)&=3(N-2)|C|.
  \end{align*}
 \item[(a3)] Given that $\sum_{i=1}^N |C|_i=2|C|$ and 
 \begin{align*}
     \sum_{i=1}^N\sum_{j\neq i}\bar{A}_C[ij]|C|_j&=\sum_{j=1}^N |C|_j\sum_{i\neq j}\bar{A}_C[ij]=\sum_{j=1}^N |C|_j^2, 
 \end{align*}
 we conclude that
 \begin{align*}
     \sum_{i=1}^N \left(
     \frac{|C|_i(|C|_i-3)}{2}
     +\sum_{j\neq i}\bar{A}_C[ij]|C|_j\right)&= \frac{1}{2}\sum_{i=1}^N |C|_i^2
     -3|C|+\sum_{i=1}^N|C|_i^2\\
     &=\frac{3}{2}\left(\sum_{i=1}^N|C|_i^2-2|C|\right).
 \end{align*}
\item[(a4)] As in \eqref{eq:T_Ti}, we have $\sum_{i=1}^N T_i(C)= 3T(C)$.
\end{itemize}
By replacing the results in (a1)--(a4) into \eqref{eq:ET_sums}, we get
\begin{align*}
    \mathbb{E}[T(\tilde G)]=&\alpha^3\binom{N}{3}+\alpha^2(1-2\alpha)(N-2)|C|+\frac{1}{2}
    \alpha\left(1-2\alpha\right)^2\left(\sum_{i=1}^N|C|_i^2-2|C|\right)+\left(1-2\alpha\right)^3T(C).
\end{align*}
    \item[(b)]
    As for part (b), it is immediate to check that: 
\begin{itemize}
    \item[(b1)] 
        $\sum_{i=1}^N\binom{N-1}{2}=3\binom{N}{3}$;
     \item[(b2)] $\sum_{i=1}^N |C|_i=2|C|$.
\end{itemize}
We then conclude that
 \begin{equation*}
        \mathbb{E}[P^{(2)}(\tilde G)] =
        3\alpha^2 \binom{N}{3}
        + (1-2\alpha)(2(N-1)\alpha-1)|C|
        + \frac{1}{2} (1-2\alpha)^2 \sum_{i=1}^N |C|_i^2.
    \end{equation*}
\end{itemize}
\end{proof}

\section{GCER distribution: proofs}\label{sec:proof_S2}
\begin{proof}[Proof of \Autoref{prop:GCERmonotone}]
According to \eqref{eq:CER2}, 
for any $i<j$,
\begin{align*}
\Pr(A_{\tilde G}[ij]=g) & = 
\left[1+\left(\frac{\alpha}{1-\alpha}\right)^{2A_W[ij]-1} \right]^{-g} \left\{1- 
\left[1+\left(\frac{\alpha}{1-\alpha}\right)^{2A_W[ij]-1} \right]^{-1}
\right\}^{1-g} \\ 
& =
\left[1+\left(\frac{\alpha}{1-\alpha}\right)^{2{A}_W[ij]-1} \right]^{-g} 
\left[
\frac{\left(\frac{\alpha}{1-\alpha}\right)^{2{A}_W[ij]-1}}{1+\left(\frac{\alpha}{1-\alpha}\right)^{2{A}_W[ij]-1}}
\right]^{1-g} \\ 
& = \left(\frac{\alpha}{1-\alpha}\right)^{ g \left(1-2{A}_W[ij]\right)} \left[
\frac{\left(\frac{\alpha}{1-\alpha}\right)^{2{A}_W[ij]-1}}{1+\left(\frac{\alpha}{1-\alpha}\right)^{2{A}_W[ij]-1}}
\right],
\end{align*}
for $g \in \{ 0,1\}$.
It follows that
\begin{align*}
    p_{\GCER}(G;W,\alpha) = \left(\frac{\alpha}{1-\alpha}\right)^{\sum_{i<j} A_G[ij] \left(1-2{A}_W[ij]\right) } \frac{1}{\prod_{i<j} 
    \left[
    1+\left(\frac{\alpha}{1-\alpha}\right)^{1-2{A}_W[ij]}
    \right] }.
\end{align*}
Then, we can evaluate the following ratio as
\begin{align*}
    &\frac{p_{\GCER}(G_2;W,\alpha)}{p_{\GCER}(G_1;W,\alpha)}   = \frac{\left(\frac{\alpha}{1-\alpha}\right)^{\sum_{i<j} A_{G_2}[ij] \left(1-2{A}_W[ij]\right) } }{
    \left(\frac{\alpha}{1-\alpha}\right)^{\sum_{i<j} A_{G_1}[ij] \left(1-2{A}_W[ij]\right) }  
    } = \left(\frac{\alpha}{1-\alpha}\right)^{\sum_{i<j}
    \left[
    A_{G_2}[ij] \left(1-2{A}_W[ij]\right) -  A_{G_1}[ij] \left(1-2{A}_W[ij]\right) \right] } \\
    &\qquad = \left(\frac{\alpha}{1-\alpha}\right)^{
    \sum_{i<j} \left[
    A_{G_2}[ij] - 2 
    A_{G_2}[ij] {A}_W[ij] - 
    A_{G_1}[ij] + 2 
    A_{G_1}[ij] {A}_W[ij] + {A}_W[ij] - {A}_W[ij] 
    \right]
    } \\
    &\qquad = \left(\frac{\alpha}{1-\alpha}\right)^{
    \sum_{i<j} \left\{ \left[ {A}_W[ij] + 
    A_{G_2}[ij] \left( 1- 2 {A}_W[ij] \right) \right] - 
    \left[ {A}_W[ij] + 
    A_{G_1}[ij] \left(  1 - 2  
    {A}_W[ij] \right) \right]
    \right\}
    } \\
 &\qquad = \left(\frac{\alpha}{1-\alpha}\right)^{
    \sum_{i<j} \left[ | A_{G_2}[ij] - {A}_W[ij] |  - 
     | A_{G_1}[ij] - {A}_W[ij] | 
    \right] } \\
    &\qquad = \left(\frac{1-\alpha}{\alpha}\right)^{d_\mathrm{T}(G_1,W)- d_\mathrm{T}(G_2,W)}, 
\end{align*}
where we used the identity $|g-w|=w+g(1-2w)$, for $g\in \left\{0, 1\right\}$ and $w \in \left[ 0, 1\right]$.\\
Then,
\begin{align*}
    \frac{p_{\GCER}(G_2;W,\alpha)}{p_{\GCER}(G_1;W,\alpha)}  
    \geq 
     1 &\Longleftrightarrow \left(d_\mathrm{T}(G_1,W)- d_\mathrm{T}(G_2,W)\right) \log\left(\frac{1-\alpha}{\alpha}\right) \geq 
     0\\
    &
     \Longleftrightarrow d_\mathrm{T}(G_1,W) 
    \geq 
    d_\mathrm{T}(G_2,W).   
\end{align*}
\end{proof}

\begin{proof}[Proof of \Autoref{prop:GCER_mode}] 
The Fréchet function, under the Hamming distance, for a random graph \(\tilde G\sim \GCER(W,\alpha)\), is defined as
\[
\Psi(L)
=
\mathbb{E}\left[d_{\mathrm H}^2(\tilde G,L)\right],
\qquad L\in\mathcal{G}_{\mathcal V},
\]
where the expectation is taken with respect to the distribution of $\tilde{G}$. Since \(\tilde G\) and \(L\) are binary graphs and edges are independent
\begin{align*}
    d_{\mathrm H}^2(\tilde G,L)
& =\left(
\sum_{i<j} \left|A_{\tilde G}[ij]-A_L[ij] \right| \right)^2 \\ & =
\sum_{i<j}\left( \left|A_{\tilde G}[ij]-A_L[ij] \right| \right)^2 + 2 \sum_{(i,j)\prec (k,\ell)} \left|A_{\tilde G}[ij]-A_L[ij]\right| \left|A_{\tilde G}[k\ell]-A_L[k\ell]\right| \\&
=
\sum_{i<j}\left|A_{\tilde G}[ij]-A_L[ij]\right| + 2 \sum_{(i,j)\prec (k,\ell)} \left|A_{\tilde G}[ij]-A_L[ij]\right| \left|A_{\tilde G}[k\ell]-A_L[k\ell]\right|,
\end{align*}
where $\prec$ denotes any fixed ordering of the node pairs. Hence,
\[
\Psi(L)
=
\sum_{i<j}
\mathbb{E}\left[
\left|A_{\tilde G}[ij]-A_L[ij]\right|
\right] +  2 \sum_{(i,j)
\prec(k,\ell)}  \mathbb{E}\left[ \left|A_{\tilde G}[ij]-A_L[ij]\right|\right] \mathbb{E}\left[ \left|A_{\tilde G}[k\ell]-A_L[k\ell]\right|\right].
\]
Although \(\Psi(L)\) is not additive over the edges, the minimizing value
of each \(A_L[ij]\) can be determined by comparing its two possible values
while holding all other entries of \(A_L\) fixed.\\
Let
\[
p_{ij}
=
\Pr(A_{\tilde G}[ij]=1
)
=
\left[
1+
\left(\frac{\alpha}{1-\alpha}\right)^{2A_W[ij]-1}
\right]^{-1},
\]
and define $Z_{k\ell}=\mathbbm 1_{\{A_{\tilde G}[k\ell]\neq A_L[k\ell]\}}$. Denote $q_{k\ell} = \mathbb{E}\left[ \left| A_{\tilde G}[k\ell]- A_L[k\ell] \right| \right]=\Pr(Z_{k\ell}=1)$.\\
For a fixed pair \(i<j\), hold all entries \(A_L[k\ell]\), with
\((k,\ell)\neq(i,j)\), fixed. Then:
\begin{itemize}
    \item[(a)]  \begin{align*}
    \mathbb{E}\left[
\left|A_{\tilde G}[ij]-A_L[ij]\right|
\right] = 
\begin{cases}
    p_{ij}, & \text{if } A_L[ij]=0, \\
    1-p_{ij}, & \text{if } A_L[ij]=1;
\end{cases}
\end{align*}
\item[(b)] 
\begin{align*}
    & 2 \mathbb{E}\left[ \left|  A_{\tilde G}[ij]-A_L[ij]\right| \right] 
    \sum_{(k,\ell) \neq (i,j)}  \mathbb{E}\left[ \left| A_{\tilde G}[k\ell]-A_L[k\ell]\right| \right]  \\ & = 
    \begin{cases}
 2 \mathbb{E}\left[ A_{\tilde G}[ij]\right] 
 \sum_{(k,\ell) \neq (i,j)} \Pr\left(Z_{k\ell}=1\right) 
 , & \text{if } A_L[ij]=0, \\
    2 \mathbb{E}\left[ 1-A_{\tilde G}[ij]\right]  \sum_{(k,\ell) \neq (i,j)} \Pr\left(Z_{k\ell}=1\right) , & \text{if } A_L[ij]=1,
   \end{cases} \\ 
   & = 
   \begin{cases}
2 p_{ij}
 \sum_{(k,\ell) \neq (i,j)} q_{k\ell} 
 , & \text{if } A_L[ij]=0, \\
    2 \left(1-p_{ij} \right) \sum_{(k,\ell) \neq (i,j)}  q_{k\ell}, & \text{if } A_L[ij]=1.
   \end{cases} 
\end{align*}
\end{itemize}
Therefore, the contribution of edge \((i,j)\) to the Fréchet function 
\begin{align*}
    & \mathbb{E}\left[ \left| A_{\tilde G}[ij]-A_L[ij]\right|\right] + 2 
    \mathbb{E}\left[ \left| A_{\tilde G}[ij]-A_L[ij]\right|\right] 
    \sum_{(k,\ell) \neq (i,j)}  \mathbb{E}\left[ \left| A_{\tilde G}[k\ell]-A_L[k\ell]\right| \right]  \\ & =
\begin{cases}
 p_{ij} + 2 p_{ij}
 \sum_{(k,\ell) \neq (i,j)} q_{k\ell}, & \text{if } A_L[ij]=0, \\
  1-p_{ij} + 2 \left( 1-p_{ij}  \right)
 \sum_{(k,\ell) \neq (i,j)} q_{k\ell}, & \text{if } A_L[ij]=1,
\end{cases} \\ & =
\begin{cases}
 p_{ij} \left( 1 + 2 \sum_{(k,\ell) \neq (i,j)} q_{k\ell}\right), & \text{if } A_L[ij]=0, \\
  \left( 1-p_{ij} \right)  \left( 1 + 2 \sum_{(k,\ell) \neq (i,j)} q_{k\ell}\right), & \text{if } A_L[ij]=1,
\end{cases}
\end{align*}
is minimized by setting \(A_L[ij]=1\) if \(p_{ij}>1/2\), by setting \(A_L[ij]=0\) if \(p_{ij}<1/2\), and by either choice if \(p_{ij}=1/2\), since the term $\left( 1 + 2 \sum_{(k,\ell) \neq (i,j)} q_{k\ell}\right)$ is strictly positive and, with all other entries of \(A_L\) fixed, this quantity does not depend on the choice of \(A_L[ij]\).\\
Since \(\alpha\in(0,1/2)\), we have
\[
p_{ij}>\frac12
\quad\Longleftrightarrow\quad
\left(\frac{\alpha}{1-\alpha}\right)^{2A_W[ij]-1}<1
\quad\Longleftrightarrow\quad
A_W[ij]>\frac12.
\]
Similarly, \(p_{ij}<1/2\) if and only if \(A_W[ij]<1/2\), while \(p_{ij}=1/2\) if and only if \(A_W[ij]=1/2\). It follows that
\[
\mathscr{F}(\tilde{G})
=
\left\{
L\in\mathcal{G}_{\mathcal V}
:
A_L[ij]=\mathbbm{1}_{\{A_W[ij]>1/2\}}
\text{ for all } i<j \text{ such that } A_W[ij]\neq 1/2
\right\}.
\]

It remains to show that this set coincides with the set of modes. Since the edges of \(\tilde G\) are independent under the \(\GCER(W,\alpha)\) distribution, the probability mass function is maximized by choosing, for each pair \(i<j\), the most likely value of \(A_{\tilde G}[ij]\). This value is \(1\) if \(p_{ij}>1/2\), \(0\) if \(p_{ij}<1/2\), and either \(0\) or \(1\) if \(p_{ij}=1/2\). By the equivalences above, the resulting set of graphs is exactly \(\mathscr{F}(\tilde G)\). Hence, the Fréchet mean set coincides with the set of modes.
\end{proof}

\begin{proof}[Proof of \Autoref{prop:entropy}]
We investigate how the entropy of the $\GCER(W,\alpha)$ distribution, defined as 
\begin{equation}\label{eq:gcer_entropy}
    H(\tilde G)=
    - \sum_{G \in \mathcal{G}_{\mathcal{V}}}
   p_{\GCER}(G;W,\alpha)
    \log \left(p_{\GCER}(G;W,\alpha)\right),
\end{equation}
depends on $\alpha$. We use the standard fact that, if $X$ and $Y$ are independent
random variables, then the entropy of their joint distribution is the sum of the
marginal entropies, namely $H(X,Y)=H(X)+H(Y)$ \citep[see, e.g.,][]{Cov99}. 
A random graph $\tilde G\sim\GCER(W,\alpha)$ can be represented as a random vector of dimension \(M\), whose entries are independent and distributed according to \eqref{eq:CER2}. 
Thus the entropy in \eqref{eq:gcer_entropy} can be written as
\begin{align}\label{eq:gcer_ij_entropy}
H(\tilde G) & = \sum_{i<j} H(A_{\tilde G}[ij]),
\end{align}
where
\begin{align*}
H(A_{\tilde G}[ij])  & = - \sum_{g \in \{ 0,1\}} \Pr(A_{\tilde G}[ij]=g
) \log\left( 
\Pr(A_{\tilde G}[ij]=g)
\right) \notag \\
& = -\frac{\left(\frac{\alpha}{1-\alpha}\right)^{2{A}_W[ij]-1}}{1+\left(\frac{\alpha}{1-\alpha}\right)^{2{A}_W[ij]-1}} \left(2{A}_W[ij]-1\right) \log\left( \frac{\alpha}{1-\alpha}\right) + \log\left( 1+\left(\frac{\alpha}{1-\alpha}\right)^{2{A}_W[ij]-1}\right), \notag
\end{align*}
is the entropy of the random edge for pair $i
< j\in\mathcal{V}$.\\
It is immediate to verify that the entropy converges to its lower bound as \(\alpha\to 0\), namely, 
\begin{align*}
\begin{cases}
     H(A_{\tilde G}[ij]) \to 0, & \text{whenever } 
     A_W[ij]\neq 1/2, \\
     H(A_{\tilde G}[ij]) \to \log(2), & \text{whenever } 
      A_W[ij]= 1/2.
\end{cases}
\end{align*}
Hence, by \eqref{eq:gcer_ij_entropy}, 
\[
H(\tilde G) \to M^*\log(2),
\]
as $\alpha\to 0$, which proves part i).\\
Conversely,  $H(A_{\tilde G}[ij])\to \log(2)$ as \(\alpha\to 1/2\). Hence, by \eqref{eq:gcer_ij_entropy},
\[
H(\tilde{G}) \to M\log(2),
\]
as $\alpha\to 1/2$, which proves part ii).\\
Finally, differentiating \eqref{eq:gcer_ij_entropy} with respect to \(\alpha\) yields
\begin{align*}
\frac{\partial}{\partial \alpha} H(\tilde G)
&=
-
\frac{\log\left(\frac{\alpha}{1-\alpha}\right)}{(1-\alpha)^2}
\sum_{i<j}
\frac{\left(2A_W[ij]-1\right)^2}
{\left\{1+\left(\frac{\alpha}{1-\alpha}\right)^{2A_W[ij]-1}\right\}^2}
\left(\frac{\alpha}{1-\alpha}\right)^{2(A_W[ij]-1)},
\end{align*}
which is non-negative 
for all \(\alpha\in(0,1/2)\),
with strict inequality unless \(A_W[ij]=1/2\) for all \(i<j\).
\end{proof}

\section{Details on Bayesian modeling and computation}\label{sec:sm_sec3}
We provide additional details on the scaling function $\tau_{n_0}$, as well as on the derivation of the full conditional distributions in Propositions~\ref{prop:fc_Gm} and \ref{prop:fc_alpha} and of the marginal likelihood in \eqref{eq:marginal-likelihood}, all presented in \Autoref{sec:model}. 

\subsection{Scaling function $\tau_{n_0}$}\label{app:app_scaling}
\Autoref{fig:scaling}
illustrates the scaling function $\tau_{n_0}$ defined in \Autoref{sec:model}, for various values of the dispersion parameter $\alpha$.

\begin{figure}[h!]
\centering
\includegraphics[width=0.5\linewidth]{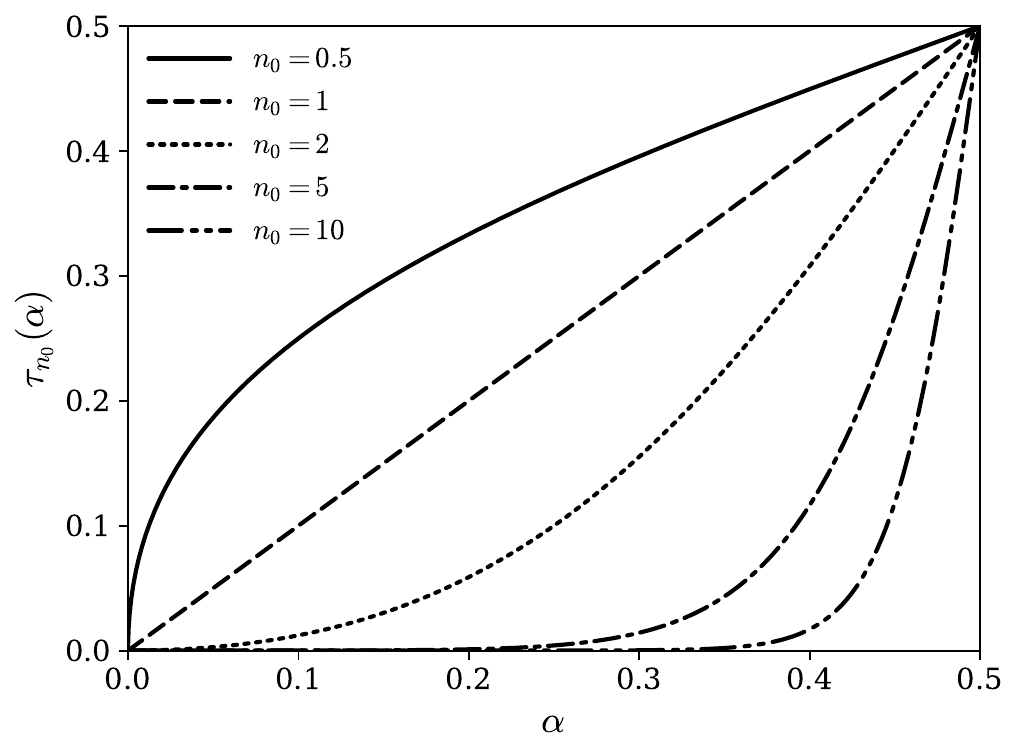}
    \caption{Scaling function $\tau_{n_0}(\alpha)$, for $\alpha\in(0,1/2)$ and $n_0\in\{0.5,1,2,5,10\}$.}\label{fig:scaling}
\end{figure}

\subsection{Full conditional distribution of the location network}

Let $s_{ij}=\sum_{\ell=1}^nA_{G_\ell}[ij]$ and
$A_{\bar G}[ij]=s_{ij}/n$ and put
$q=\alpha/(1-\alpha)$. Up to terms not depending on $A_C[ij]$, the
edgewise likelihood and prior are proportional to
\begin{equation*}
q^{nA_C[ij](1-2A_{\bar G}[ij])}
\quad\text{and}\quad
q^{n_0A_C[ij](1-2A_W[ij])},
\end{equation*}
respectively. Their product is
\begin{equation*}
q^{(n_0+n)A_C[ij](1-2A_{W'}[ij])},
\qquad
A_{W'}=\frac{n_0A_W+nA_{\bar G}}{n_0+n}.
\end{equation*}
Normalizing edge by edge proves the GCER full conditional for $C$ stated in \Autoref{prop:fc_Gm}. 
We also note that, when $A_W[ij]=1/2$ for every edge,
$(n_0+n)\{1-2A_{W'}[ij]\}=n-2s_{ij}$, so the induced full
conditional distribution is independent of $n_0$.

\subsection{Full conditional distribution of the dispersion parameter}

Let $D(C)=\sum_{\ell=1}^nd_{\mathrm H}(G_\ell,C)$ and
\begin{equation*}
S(C)=D(C)+n_0\sum_{i<j}A_C[ij]\{1-2A_W[ij]\}.
\end{equation*}
For a general prior location $W$, the full conditional density is
\begin{align}
p(\alpha\mid G^{(1:n)},C)
\propto{}&
\alpha^{a+S(C)-1}(1-\alpha)^{b+nM-S(C)-1}\notag\\
&\times\prod_{i<j}
\left[1+
\left\{\frac{\alpha}{1-\alpha}\right\}^{n_0(1-2A_W[ij])}
\right]^{-1},
\qquad 0<\alpha<\frac12.
\label{eq:general-alpha-full-conditional}
\end{align}
When $A_W[ij]=1/2$ for every edge, the product in
\eqref{eq:general-alpha-full-conditional} is constant and $S(C)=D(C)$.
The density therefore reduces to the truncated Beta full conditional in \Autoref{prop:fc_alpha}.

\subsection{Marginal likelihood}\label{app:app_ml}

Set $r_{ij}=n_0\{1-2A_W[ij]\}$ and
$q=\alpha/(1-\alpha)$. The GCER prior implies
\begin{equation*}
\Pr\{A_C[ij]=1\mid\alpha\}=\frac{q^{r_{ij}}}{1+q^{r_{ij}}},
\qquad
\Pr\{A_C[ij]=0\mid\alpha\}=\frac{1}{1+q^{r_{ij}}}.
\end{equation*}
Conditional on $A_C[ij]=0$, the edgewise likelihood is
$\alpha^{s_{ij}}(1-\alpha)^{n-s_{ij}}$, whereas, conditional on
$A_C[ij]=1$, it is
$\alpha^{n-s_{ij}}(1-\alpha)^{s_{ij}}$. Summing over the two possible
values of $A_C[ij]$ and using independence across edges therefore gives the marginal likelihood in \eqref{eq:marginal-likelihood}:
\begin{align}
L(\alpha;G^{(1:n)})
={}&\prod_{i<j}
\frac{
\alpha^{s_{ij}}(1-\alpha)^{n-s_{ij}}
+q^{r_{ij}}\alpha^{n-s_{ij}}(1-\alpha)^{s_{ij}}
}{1+q^{r_{ij}}}\notag\\
={}&\prod_{i<j}
\frac{
\alpha^{s_{ij}}(1-\alpha)^{n-s_{ij}+r_{ij}}
+\alpha^{n-s_{ij}+r_{ij}}(1-\alpha)^{s_{ij}}
}{
\alpha^{r_{ij}}+(1-\alpha)^{r_{ij}}
}.
\label{eq:supp-marginal-likelihood}
\end{align}
When $W=W_0$, $r_{ij}=0$ for every edge, and
\eqref{eq:supp-marginal-likelihood} reduces to
\begin{equation}\label{eq:mar_lik_conj}
L(\alpha;G^{(1:n)})
={} 2^{-M}\prod_{i<j}\left\{
\alpha^{s_{ij}}(1-\alpha)^{n-s_{ij}}
+\alpha^{n-s_{ij}}(1-\alpha)^{s_{ij}}
\right\}.
\end{equation}

\subsection{Uniqueness of the maximizer of the marginal likelihood}\label{sec:unique}
We next characterize the existence and uniqueness of an interior maximizer
of the marginal likelihood under the semi-conjugate specification.
\begin{proposition}\label{prop:unique_max}
Assume \(n\geq2\). For each \(i<j\), let
\[
s_{ij}=\sum_{\ell=1}^n A_{G_\ell}[ij],
\qquad
k_{ij}=\min\{s_{ij},n-s_{ij}\},
\qquad
d_{ij}=n-2k_{ij}=|n-2s_{ij}|,
\]
and define
\[
K=\sum_{i<j}k_{ij},
\qquad
D=\sum_{i<j}\binom{d_{ij}}{2}.
\]
Under the semi-conjugate specification, the marginal likelihood
\(L(\alpha;G^{(1:n)})\) has a unique maximizer in \((0,1/2)\) if and only if
\[
0<K<D.
\]
\end{proposition}
\begin{remark}
The condition \(K>0\) requires at least one edge to change its state across
the observed networks. Moreover, since \(2k_{ij}=n-d_{ij}\),
\[
D-K
=
\frac{1}{2}
\left(
\sum_{i<j}d_{ij}^2-nM
\right).
\]
Thus,
\[
D>K
\quad\Longleftrightarrow\quad
\frac{1}{M}\sum_{i<j}d_{ij}^2>n,
\]
meaning that the squared edge-frequency imbalances must, on average,
exceed \(n\).
\end{remark}
\begin{proof}[Proof of Proposition~\ref{prop:unique_max}]
Under the semi-conjugate specification, that is, \(W=W_0\), the marginal
likelihood in \eqref{eq:marginal-likelihood} reduces to
\[
L(\alpha;G^{(1:n)})
=
2^{-M}
\prod_{i<j}
\left\{
\alpha^{s_{ij}}(1-\alpha)^{n-s_{ij}}
+
\alpha^{n-s_{ij}}(1-\alpha)^{s_{ij}}
\right\}.
\]
Consider the reparametrization
\[
\rho=(1-2\alpha)^2,
\qquad
\alpha=\frac{1-\sqrt{\rho}}{2},
\qquad
\rho\in(0,1).
\]
The map \(\alpha\mapsto\rho\) is a strictly decreasing bijection from
\((0,1/2)\) to \((0,1)\), so maximizing the marginal likelihood with
respect to \(\alpha\) is equivalent to maximizing it with respect to
\(\rho\). For each pair \(i<j\),
\begin{equation}\label{eq:id1}
\begin{aligned}
\frac{
\alpha^{s_{ij}}(1-\alpha)^{n-s_{ij}}
+
\alpha^{n-s_{ij}}(1-\alpha)^{s_{ij}}
}{2}
=
\frac{1}{2}
\bigl[\alpha(1-\alpha)\bigr]^{k_{ij}}
\left\{
\alpha^{d_{ij}}+(1-\alpha)^{d_{ij}}
\right\}.
\end{aligned}
\end{equation}
We define
\[
P_d(\rho)
=
\sum_{h=0}^{\lfloor d/2\rfloor}
\binom{d}{2h}\rho^h
\]
and observe that \(4\alpha(1-\alpha)=1-\rho\) and, for any nonnegative
integer \(d\),
\[
\alpha^d+(1-\alpha)^d
=
2^{1-d}
\sum_{h=0}^{\lfloor d/2\rfloor}
\binom{d}{2h}\rho^h.
\]
Thus, using \(n=2k_{ij}+d_{ij}\), we rewrite \eqref{eq:id1} as
\[
\frac{
\alpha^{s_{ij}}(1-\alpha)^{n-s_{ij}}
+
\alpha^{n-s_{ij}}(1-\alpha)^{s_{ij}}
}{2}
=
2^{-n}(1-\rho)^{k_{ij}}P_{d_{ij}}(\rho).
\]
Therefore,
\[
L\left(\frac{1-\sqrt{\rho}}{2};G^{(1:n)}\right)
=
2^{-nM}
\prod_{i<j}
(1-\rho)^{k_{ij}}P_{d_{ij}}(\rho).
\]
Dropping the additive constant \(-nM\log(2)\), denote the corresponding
log marginal likelihood by
\begin{equation}\label{eq:log-likelihood-rho}
\Lambda(\rho)
=
K\log(1-\rho)
+
\sum_{i<j}\log P_{d_{ij}}(\rho).
\end{equation}
It can be verified that, for \(d\geq2\), the \(\lfloor d/2\rfloor\)
roots of \(P_d\) are
\[
a_{dr}
=
-\tan^2\left(\frac{(2r-1)\pi}{2d}\right)<0,
\qquad
r=1,\ldots,\left\lfloor\frac d2\right\rfloor,
\]
and are therefore real and strictly negative. Since \(P_d(0)=1\),
\[
P_d(\rho)
=
\prod_{r=1}^{\lfloor d/2\rfloor}
\left(1-\frac{\rho}{a_{dr}}\right),
\qquad
\frac{\mathrm d^2}{\mathrm d\rho^2}\log P_d(\rho)
=
-\sum_{r=1}^{\lfloor d/2\rfloor}
\frac{1}{(\rho-a_{dr})^2}.
\]
For \(d\in\{0,1\}\), \(P_d(\rho)=1\); hence, the factorization and
second-derivative formula above remain valid under the usual conventions
for empty products and sums. Differentiating
\eqref{eq:log-likelihood-rho} and using the factorization above gives
\begin{align*}
\Lambda'(\rho)
&=
-\frac{K}{1-\rho}
+
\sum_{i<j}
\frac{P_{d_{ij}}'(\rho)}{P_{d_{ij}}(\rho)},\\
\Lambda''(\rho)
&=
-\frac{K}{(1-\rho)^2}
-
\sum_{i<j}
\sum_{r=1}^{\lfloor d_{ij}/2\rfloor}
\frac{1}{(\rho-a_{d_{ij}r})^2}.
\end{align*}
Thus, if \(K>0\), then \(\Lambda''(\rho)<0\) for every
\(\rho\in(0,1)\), so \(\Lambda\) is strictly concave on \((0,1)\).
Moreover, since \(P_d(0)=1\) and \(P_d'(0)=\binom{d}{2}\),
\(\lim_{\rho\to0^+}\Lambda'(\rho)=D-K\). If \(K>0\), then
\(\lim_{\rho\to1^-}\Lambda'(\rho)=-\infty\), because the term
\(-K/(1-\rho)\) diverges to \(-\infty\), while the remaining terms stay
finite. If \(0<K<D\), the strictly decreasing function \(\Lambda'\)
therefore has a unique zero \(\rho^\star\in(0,1)\), which is the unique
maximizer of \(\Lambda\) on \((0,1)\).

Conversely, if \(K>0\) and \(D\leq K\), then
\(\lim_{\rho\to0^+}\Lambda'(\rho)=D-K\leq0\). Since \(\Lambda'\) is
strictly decreasing, \(\Lambda'(\rho)<0\) for every \(\rho\in(0,1)\).
If \(K=0\), then \(k_{ij}=0\) and hence \(d_{ij}=n\) for every \(i<j\).
Since \(n\geq2\), \(P_n'(\rho)>0\) on \((0,1)\), so \(\Lambda\) is
strictly increasing. Thus, in either case, no interior maximizer exists.

Therefore, an interior maximizer exists if and only if \(0<K<D\), and is
then unique. Since \(\alpha\mapsto\rho=(1-2\alpha)^2\) is a bijection from
\((0,1/2)\) to \((0,1)\), the corresponding unique maximizer is
\[
\hat\alpha=\frac{1-\sqrt{\rho^\star}}{2}.
\]
\end{proof}

\section{Additional numerical experiments}\label{sec:sm_sec4}
We extend the simulation study in \Autoref{sec:application} of the main article by considering additional values of $\alpha \in \{0.1,0.3\}$. For all scenarios, we assess the accuracy of the proposed estimators using the normalized Hamming error for $C$ and the absolute error for $\alpha$. Thus, in the extended numerical experiment, we consider settings spanning
$n\in\{10,100,1000\}$, $N\in\{10,100,1000,2000\}$, and
$\alpha\in\{0.1,0.2,0.3\}$. For each pair $(N,\alpha)$, the location parameter
$C$ was generated from a two-community stochastic block model with
within-community edge probability $0.8$ and between-community edge probability
$0.1$, with nodes assigned to the two communities independently and with equal probability. Conditional on $C$, the observations were independently generated from
$\CER(C,\alpha)$. For each of the resulting $36$ settings, we generated $100$ independent datasets. As in the main article, we used $W=W_0$, under which the induced prior and posterior distributions are invariant to the choice of $n_0$, and set $a=b=1$. The Gibbs sampler was run for a total of $11{,}000$ iterations, with the first $1{,}000$ discarded as burn-in and every 10th subsequent draw retained, yielding $1{,}000$ posterior draws. 
From these draws, we computed \(\hat\alpha\) as the posterior mean and
obtained \(\hat C\) by thresholding the posterior edge-inclusion
probabilities at \(1/2\).
The normalized Hamming error for $C$ and the absolute error for $\alpha$, reported in Figures~\ref{fig:disc-C} and~\ref{fig:disc-alpha}, respectively, show that estimates from both the EB and Gibbs methods increasingly concentrate around the data-generating values as $n$ grows. The largest variability is observed for $N=10$ and $n=10$, and decreases rapidly with increasing $n$.
\begin{figure}[h!]
\centering
\includegraphics[width=\linewidth]{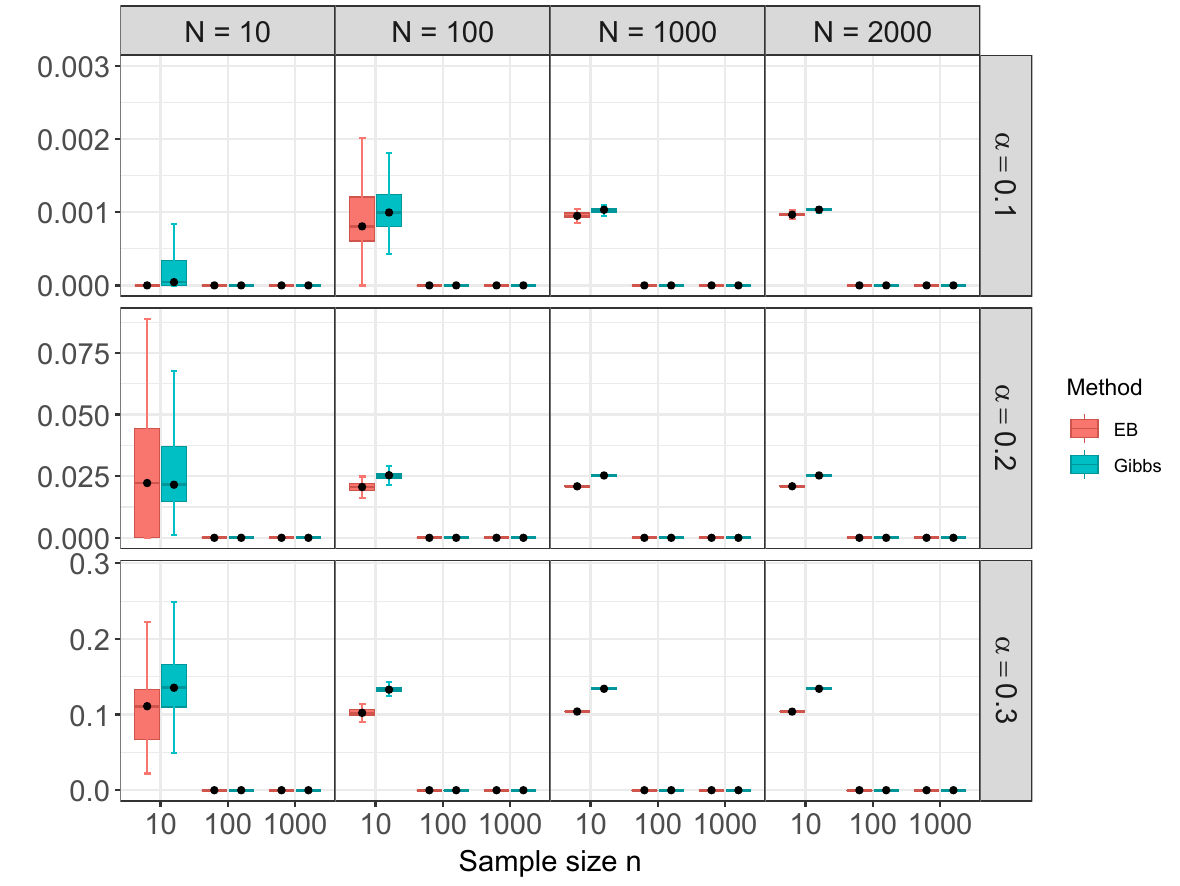}
\caption{Boxplots of the normalized Hamming error
$d_{\mathrm H}(C,\hat C)/M$, based on 100 replications for each simulation
setting; black dots mark the median errors.}
\label{fig:disc-C}
\end{figure}
\begin{figure}[h!]
\centering
\includegraphics[width=\linewidth]{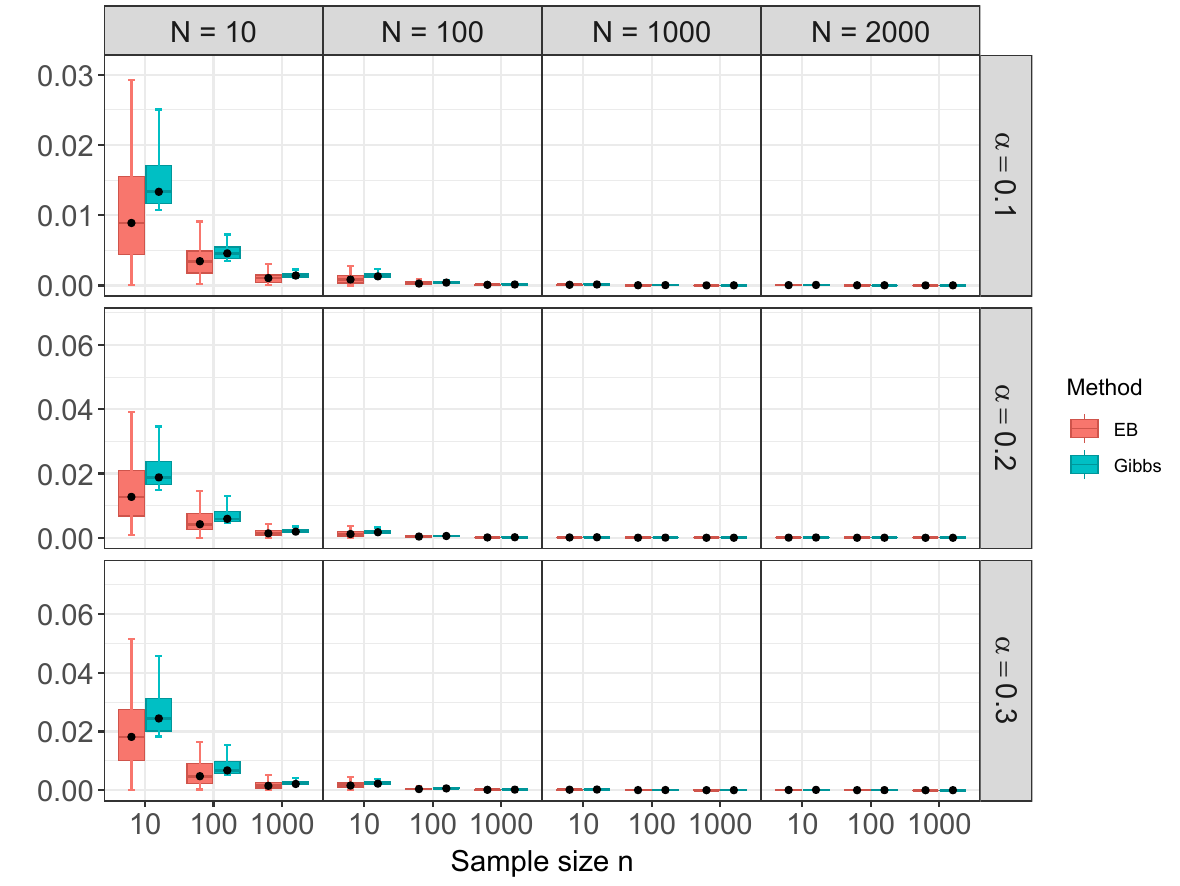}
\caption{Boxplots of the absolute error $|\alpha-\hat\alpha|$, based on 100
replications for each simulation setting; black dots mark the median errors.}
\label{fig:disc-alpha}
\end{figure}
For the additional values of $\alpha$ considered here, \Autoref{fig:ecc-full} displays the estimates of the approximate expected global clustering coefficient introduced in \Autoref{prop:gcc}, while \Autoref{fig:edge-density-full} displays the estimates of the expected edge density in \eqref{eq:edge-density}. The results lead to conclusions analogous to those in \Autoref{sec:application}, 
irrespective of the value of $\alpha$.
\begin{figure}[h]
\centering
\includegraphics[width=\linewidth]{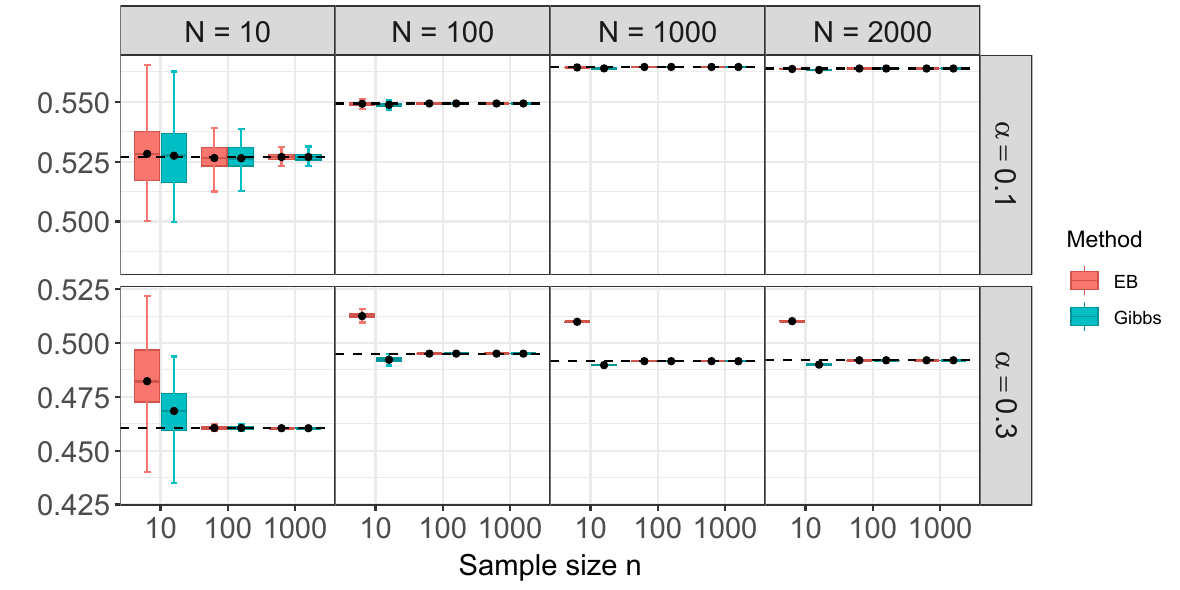}
\caption{Approximate expected global clustering coefficient for
$\alpha\in\{0.1,0.3\}$. Boxplots summarize 100 replications for each
simulation setting; dashed black lines indicate the data-generating values,
and black dots mark the median estimates.}
\label{fig:ecc-full}
\end{figure}
\begin{figure}[h]
\centering
\includegraphics[width=\linewidth]{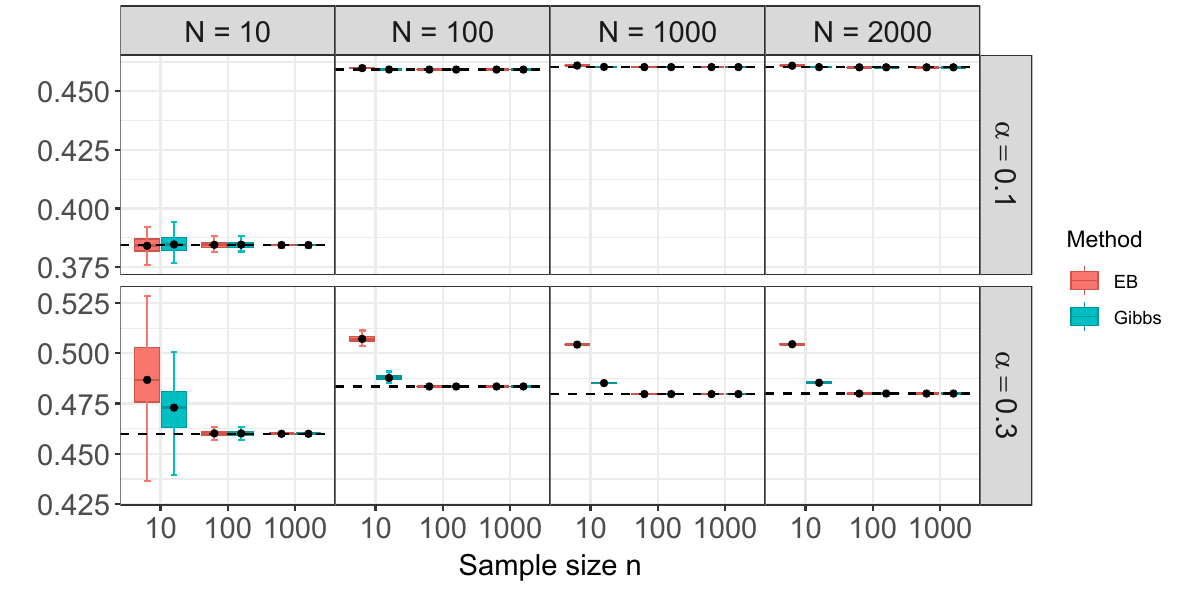}
\caption{Expected edge density for $\alpha\in\{0.1,0.3\}$. Boxplots summarize
100 replications for each simulation setting; dashed black lines indicate
the data-generating values, and black dots mark the median estimates.}
\label{fig:edge-density-full}
\end{figure}
Finally, 
\Autoref{tab:computational_time} reports the computational time, including the computation of the network summaries,
for all settings with $n=1000$. Since both methods depend on the observed graphs only through the edge-by-edge average of the adjacency matrices,
$A_{\overline{G}}=n^{-1}\sum_{\ell=1}^n A_{G_\ell}$,
their computational cost is independent of $n$ once $A_{\overline{G}}$ has been computed.
\begin{table}[h]
\centering
\begin{tabular}[t]{rrrll}
\toprule
n & N & alpha & EB (IQR) & Gibbs (IQR)\\
\midrule
1000 & 10 & 0.1 & 0.00 (0.00) & 0.71 (0.03)\\
1000 & 100 & 0.1 & 0.01 (0.00) & 8.44 (0.55)\\
1000 & 1000 & 0.1 & 0.55 (0.19) & 1020.42 (32.88)\\
1000 & 2000 & 0.1 & 2.09 (0.39) & 3404.55 (71.47)\\
\addlinespace
1000 & 10 & 0.2 & 0.00 (0.00) & 0.71 (0.04)\\
1000 & 100 & 0.2 & 0.01 (0.00) & 8.77 (0.61)\\
1000 & 1000 & 0.2 & 0.96 (0.07) & 945.96 (27.34)\\
1000 & 2000 & 0.2 & 4.11 (3.31) & 3424.13 (58.01)\\
\addlinespace
1000 & 10 & 0.3 & 0.00 (0.00) & 0.70 (0.04)\\
1000 & 100 & 0.3 & 0.02 (0.01) & 8.33 (0.29)\\
1000 & 1000 & 0.3 & 1.19 (1.22) & 931.92 (23.49)\\
1000 & 2000 & 0.3 & 6.79 (7.01) & 3397.53 (91.83)\\
\bottomrule
\end{tabular}
\caption{Median computational time (IQR) in seconds.}\label{tab:computational_time}
\end{table}
For moderate and large $n$, the two methods yield comparable
inferential results, with computational cost remaining their main difference.

\putbib
\end{bibunit}

\end{document}